\documentclass[12pt,draftcls,onecolumn]{IEEEtran}

\usepackage{cite}
\usepackage{graphicx}
\usepackage[cmex10]{amsmath} 
\usepackage{amssymb,amsthm}
\usepackage{flushend}
\usepackage{bbm}
\usepackage{enumitem}

\newtheorem{Theorem}{Theorem}
\newtheorem{Lemma}{Lemma}

\newtheorem{Definition}{Definition}

\makeatletter%
\if@twocolumn%
\newcommand{\Figwidth}{\columnwidth}%

\newcommand{\includeonetwocol}[2]{#2}
\def\twocolbreak{\nonumber\\ &}%
\def\twocolnewline{\nonumber\\}%
\def\twocolAlignMarker{&}%
\else
\newcommand{\Figwidth}{4.5in}%

\newcommand{\includeonetwocol}[2]{#1}
\def\twocolbreak{}%
\def\twocolnewline{}%
\def\twocolAlignMarker{}%
\fi%
\makeatother%

\def\TreeSideInfo{{\tilde{\boldsymbol \tau}}}
\def\TreeLabels{{\boldsymbol \tau}}
\def\RootEstimator{{\hat{\tau}_0}}
\def\RootLabel{{\tau_0}}
\def\CroppedTree{{T^t}}
\def\LLRCroppedTree{{\Gamma_0^t}}
\def\TreeIter{{\hat{t}}}
\def\TreeDepth{{t}}

\allowdisplaybreaks

\begin{document}
\title{EXIT Analysis for Community Detection}

\author{Hussein Saad, {\em Student Member, IEEE}, and Aria Nosratinia, {\em Fellow, IEEE}
\thanks{The authors are with the Department of Electrical Engineering, University of Texas at
    Dallas, Richardson, TX 75083-0688 USA, E-mail:
    hussein.saad@utdallas.edu; aria@utdallas.edu.}
}
\maketitle


\begin{abstract}
This paper employs the extrinsic information transfer (EXIT) method, a technique imported from the analysis of the iterative decoding of error control codes, to study the performance of belief propagation in community detection in the presence of side information. We consider both the detection of a single (hidden) community, as well as the problem of identifying two symmetric communities. For single community detection, this paper demonstrates the suitability of EXIT to predict the asymptotic phase transition for weak recovery. More importantly, EXIT analysis is leveraged to produce useful insights such as the performance of belief propagation near the threshold. For two symmetric communities, the asymptotic residual error for belief propagation is calculated under finite-alphabet side information, generalizing a previous result with noisy labels. EXIT analysis is used to illuminate the effect of side information on community detection, its relative importance depending on the correlation of the graphical information with node labels, as well as the effect of side information on residual errors.

\end{abstract}

\IEEEpeerreviewmaketitle


\section{Introduction}
\label{intro}

Detecting communities in graphs is a fundamental problem that has been studied in various fields including statistics~\cite{SBM_first,Min_max_SBM,stat1,stat2,stat3}, computer science~\cite{SCT,CS1,CS2,CS3,CS4} and theoretical statistical physics~\cite{Phys,Phys2}. Applications of community detection include finding like-minded people in social networks~\cite{social}, improving recommendation systems~\cite{recommendation}, and detecting protein complexes~\cite{protein}.

Several models are proposed for random graphs that exhibit a community structure; a survey can be found in~\cite{survey}. This paper considers the stochastic block model (SBM)~\cite{exact_general_sbm}, which is widely used as a model for community detection and as a benchmark for clustering algorithms~\cite{SBM_first,Min_max_SBM,SCT,Phys}. 
This paper addresses two models: the stochastic block model for one community~\cite{infor_limits,recovering_one}, and the binary symmetric stochastic block model~\cite{Main_paper2,exact_sbm}.  Most of the community detection literature recovers the communities from a purely graphical observation. However, in practical scenarios, extra information (side information) about the nodes might also be available. For example, social networks such as Facebook and Twitter have access to information other than the graph edges such as date of birth, nationality, school. Therefore, in this paper we consider a generalization of the standard community detection problem, where in addition to the connectivity graph, some side information related to each node's individual attributes is available for inference.

In the area of community detection, some works~\cite{exact_general_sbm,exact_sbm,Saad:JSTSP18,Saad:single,ISIT-2018-2,our2} have concentrated on deriving information theoretic limits, while some others~\cite{Main_paper2,Kadavankandy:SingleCommunity,MoselXu:ACM16,Cai:BP-BEC,ISIT-2018-1} proposed efficient algorithms such as belief propagation and semi-definite programming and study their asymptotic performance. 
This paper proposes a new tool\footnote{While EXIT analysis has been used in the context of error control codes and related subjects in communication, its introduction to the area of community detection is novel. Please note that EXIT charts~\cite{johnson} predate, and are unrelated to~\cite{EXIT_not_real} which also uses the acronym "EXIT".} for the analysis of the performance of local message passing algorithms, e.g., belief propagation, for community detection with side information. EXIT analysis has been used to understand the behavior of iterative algorithms~\cite{johnson} in the context of error control and communication systems. Instead of calculating and tracking the probability density of the estimate or its log likelihood  (density evolution) which can be complicated, EXIT analysis tracks the evolution of the mutual information (a scalar value) at each iteration of the algorithm. EXIT analysis has the advantage of modest complexity (compared with density evolution), robustness to approximation errors, and production of useful insights~\cite{johnson}. By observing the EXIT chart, one can predict whether the decoder will fail, deduce the approximate number of iterations needed to decode, as well as approximate error probability after decoding. EXIT charts also have the additional benefit of an information theoretic interpretation~\cite{johnson}.

We apply EXIT analysis to single-community detection as well as to binary symmetric community detection, each with side information, and leverage this technique to provide insights on:
\begin{itemize}

\item The effect of the quality and quantity of side information on the performance of belief propagation, e.g. probability of error

\item The asymptotic threshold for weak recovery, achieving a vanishing residual error

\item The performance of belief propagation near the optimal threshold

\item The performance of belief propagation through the first few iterations

\item Approximating the number of iterations needed for convergence

\end{itemize}

The technical distinction and novelty of this paper can be explained as follows: EXIT analysis was originally developed in the context of communication systems for bipartite graphs in which some nodes carry information while some other nodes represent the constraints on the data nodes (e.g. via parity check equations or the structure or memory of a communication channel). The present work aims to employ EXIT analysis in a scenario where the above conditions do not apply, and therefore the EXIT analysis must be developed anew for the scenario where each node in a general tree has both an individual label (information) as well as information that is applicable to other nodes. This gives rise to new EXIT equations. In other words, in the original EXIT analysis, all mutual information was calculated with respect to a subset of node labels, i.e., bit-node variables, whereas now all nodes have information. Since we are now interested in a graph that has a stochastic symmetry, the input/output belief propagation equations must be reinterpreted once again in terms of extrinsic information. This statement will be further clarified in the sequel while developing the details of EXIT equations.

An early version of this work, without considering side information and using a subset of graph system models considered herein, appeared in the following conference paper~\cite{Saad:ISIT2016}.

\section{System Models}\label{sys.}

Throughout this paper the community label of node $i$ is denoted by $x_i$, the side information of node $i$ by $y_i$, the vector of the nodes true labels by $\boldsymbol{x}^{*}$, the vector of the nodes side information by $\boldsymbol{y}$, and the observed graph by $G$. We assume that conditioned on $\boldsymbol{x}$, $G$ and $\boldsymbol{y}$ are independent. The goal is to recover $\boldsymbol{x}^{*}$ from the observation of $\boldsymbol G$ and $\boldsymbol y$. The alphabet for $y_i$ is denoted with $\{u_{1}, u_{2}, \cdots, u_{M}\}$, where $M$ is the cardinality of side information which is assumed to be bounded and constant across $n$.

Two system models are considered in this paper. The first, the binary symmetric stochastic block model, consists of $n$ nodes with $x_{i} \in \{\pm1\}$. The node labels are independent
and identically distributed across $n$, with $1$ and $-1$ labels having equal probability. Each two nodes are connected with an edge with probability $\frac{a}{n}$ if the two nodes belong to the same community and with probability $\frac{b}{n}$, otherwise, for $a > b >0$. In addition to the graph, each node independently observes side information, $y_{i}$, according to:
\begin{align}
\alpha_{+,m} \triangleq \mathbb{P}(y_{i} = u_{m} | x_{i} = 1) \\
\alpha_{-,m} \triangleq \mathbb{P}(y_{i} = u_{m} | x_{i} = -1) 
\end{align}

It is further assumed that as $n \to \infty$: $a,b \to \infty$ such that $\frac{a-b}{\sqrt{b}} = \mu$, for a fixed positive constant $\mu$ and that the average degree $\frac{(a+b)}{2} = n^{o(1)}$. The latter condition is crucial in our analysis, by enabling the approximation of the neighborhood of a given node in the graph by a tree~\cite{Main_paper1,Main_paper2}.

The second model studied in this paper is the one-community stochastic block model, consisting of $n$ nodes and containing a hidden community $C^{*}$ with size $|C^{*}| = K$. Let $x_{i} = 1$ if $i \in C^{*}$ and $x_{i} = 0$ if $i \notin C^{*}$. The underlying distribution of the graph is as follows:  an edge connects a pair of nodes with probability $p$ if both nodes are in $C^{*}$ and with probability $q$ otherwise, with $p\geq q$. For each node $i$, side information $y_{i}$ is observed according to the distribution:
\begin{align}
\alpha_{+,m} \triangleq \mathbb{P}(y_{i} = u_{m} | x_{i} = 1) \\
\alpha_{-,m} \triangleq \mathbb{P}(y_{i} = u_{m} | x_{i} = 0) 
\end{align} 

Define 
\begin{equation}
\lambda \triangleq \frac{K^{2}(p-q)^{2}}{(n-K)q}.
\end{equation}
We assume $\frac{K}{n}$, the LLR of side information and $\lambda$ are constants independent of $n$, while $nq, Kq \overset{n\rightarrow\infty}{\xrightarrow{\hspace{0.2in}}} \infty$, which implies that $\frac{p}{q} \overset{n\rightarrow\infty}{\xrightarrow{\hspace{0.2in}}} 1$. Furthermore,  $np = n^{o(1)}$.


\section{Binary Symmetric Stochastic Block Model}\label{Bel.}

Studying the performance of belief propagation with noisy-label side information was introduced in~\cite{Main_paper2}. This section generalizes the results to M-ary side information and introduces EXIT analysis as a new tool to study the performance of belief propagation for community detection.  A key idea in our analysis is the relation between inference on graphs and inference on the corresponding Galton-Watson trees~\cite{Main_paper2}. 

\begin{Definition}
For a node $i$, let $(T_{i},\tau,\tilde{\tau})$ be a Poisson two-type branching process tree rooted at $i$, where $\tau$ is a $\pm1$ labeling of  nodes in $T_{i}$. Let $\tau_{i}$ be chosen uniformly at random from $\{\pm1\}$. Each node $j$ in $T_{i}$ will have $L_{j} \sim \text{Pois}(\frac{a}{2})$ children with label $\tau_{j}$ and $M_{j}\sim \text{Pois}(\frac{b}{2})$ children with label $-\tau_{j}$. Finally, for each node $j$, an $M$-ary side information $\tilde{\tau}_{j}$ is observed according to the conditional distributions $\alpha_{+,m}$ and $\alpha_{-,m}$.
\end{Definition}

Let $T_{j}^{t}$ be the sub-tree of $T_{i}$ rooted at node $j$ with depth $t$. The problem of inference on trees with side information is to estimate the label of the root $\tau_{i}$ given observation of $(T_{i}^{t},\tilde{\tau}_{T_{i}^{t}})$, where $\tilde{\tau}_{T_{i}^{t}}$ is the side information of all the nodes in the tree rooted at $i$ with depth $t$. It then follows that the error probability for an estimator $\hat{\tau}_{i}(T_{i}^{t},\tilde{\tau}_{T_{i}^{t}})$ is: 
\begin{align*}
q_{T^{t}} &= \frac{1}{2} \mathbb{P}(\hat{\tau}_{i} \\
&= 1 | \tau_{i} = -1) + \frac{1}{2} \mathbb{P}(\hat{\tau}_{i} = -1 | \tau_{i} = 1).
\end{align*}
Let $q^{*}_{T^{t}}$ be the error probability achieved by the optimal estimator, i.e. maximum a posteriori (MAP). Note that the MAP estimator for any node $i$ can be written as: $\hat{\tau}_{MAP} = 2 \times 1_{\{\Gamma_{i}^{t} \geq 0\}} -1$, where $\Gamma_{i}^{t}$ is the log likelihood ratio and can be defined as:
\begin{equation}\label{eq.1}
\Gamma_{j}^{t} = \frac{1}{2} \log\Bigg( \frac{\mathbb{P}(T_{j}^{t},\tilde{\tau}_{T_{j}^{t}} | \tau_{j} = 1)}{\mathbb{P}(T_{j}^{t},\tilde{\tau}_{T_{j}^{t}} | \tau_{j} = -1)} \Bigg)
\end{equation}
$\forall j \in T_{i}$. The log likelihood ratio $\Gamma_{j}^{t}$ can be further computed via a recursive formula which is the basis for the belief propagation algorithm.

\begin{Lemma}
\label{Le.1}
Let  ${\mathcal N}_{j}$ denote the children of node $j$, $N_j\triangleq|{\mathcal N}_j|$, $\beta = \frac{1}{2}\log(\frac{a}{b})$ and $h_{j} \triangleq \frac{1}{2} \log\big(\frac{\mathbb{P}(\tilde{\tau}_{j} | \tau_j =1)}{\mathbb{P}(\tilde{\tau}_{j} | \tau_j = -1)}\big)$. Then, for all $t\geq 1$, 
\begin{align}
\Gamma_{j}^{t} = h_j + \frac{1}{2} \sum_{k \in {\mathcal N}_{j}} \log\Bigg(\frac{1 + e^{2\beta + 2\Gamma_{k}^{t-1}} }{e^{2\beta} + e^{2\Gamma_{k}^{t-1}}}\Bigg)
\end{align}
\end{Lemma}

\begin{proof}

\begin{align}
\Gamma_{j}^{t} =& \frac{1}{2} \log\Bigg( \frac{\mathbb{P}(T_{j}^{t},\tilde{\tau}_{T_{j}^{t}} | \tau_{j} = 1)}{\mathbb{P}(T_{j}^{t},\tilde{\tau}_{T_{j}^{t}} | \tau_{j} = -1)} \Bigg) \nonumber \\ 
 \overset{(a)}{=}& \frac{1}{2} \log\Bigg( \frac{\mathbb{P}\big( N_{j},\tilde{\tau}_{j} | \tau_{j} =1\big)}{\mathbb{P}\big( N_{j},\tilde{\tau}_{j} | \tau_{j} =-1\big)} \Bigg) \twocolbreak+  \log\Bigg( \frac{\prod_{k \in {\mathcal N}_{j}} \mathbb{P}\big(T_{k}^{t-1},\tilde{\tau}_{T_{k}^{t-1}} | \tau_{j} = 1\big)}{\prod_{k \in {\mathcal N}_{j}} \mathbb{P}\big(T_{k}^{t-1},\tilde{\tau}_{T_{k}^{t-1}} | \tau_{j} = -1\big)} \Bigg) \nonumber \\
 \overset{(b)}{=}& \frac{1}{2} \log\Bigg( \frac{\mathbb{P}\big( N_{j} | \tau_{j} =1\big)}{\mathbb{P}\big( N_{j} | \tau_{j} =-1\big)} \Bigg) + \frac{1}{2} \log\Bigg( \frac{\mathbb{P}\big( \tilde{\tau}_{j} | \tau_{j} =1\big)}{\mathbb{P}\big( \tilde{\tau}_{j} | \tau_{j} =-1\big)} \Bigg) + \nonumber \\ 
&  \frac{1}{2} \sum_{k \in {\mathcal N}_{j}} \log\Bigg( \frac{ \sum_{\tau_{k} \in \{\pm1\} } \mathbb{P}\big( T_{k}^{t-1},\tilde{\tau}_{T_{k}^{t-1}} | \tau_{k} \big) \mathbb{P}\big( \tau_{k}|\tau_{j}=1 \big) }{\sum_{\tau_{k} \in \{\pm1\} }  \mathbb{P}\big( T_{k}^{t-1},\tilde{\tau}_{T_{k}^{t-1}} | \tau_{k} \big) \mathbb{P}\big( \tau_{k}|\tau_{j}=-1 \big)} \Bigg) \nonumber \\ 
 \overset{(c)}{=}& h_j +  \twocolbreak\frac{1}{2} \sum_{k \in {\mathcal N}_{j}} \log\Bigg( \frac{ a \mathbb{P}\big( T_{k}^{t-1},\tilde{\tau}_{T_{k}^{t-1}} | \tau_{k} =1 \big) + b \mathbb{P}\big( T_{k}^{t-1},\tilde{\tau}_{T_{k}^{t-1}} | \tau_{k} =-1 \big)}{b \mathbb{P}\big( T_{k}^{t-1},\tilde{\tau}_{T_{k}^{t-1}} | \tau_{k} =1 \big) + a \mathbb{P}\big( T_{k}^{t-1},\tilde{\tau}_{T_{k}^{t-1}} | \tau_{k} =-1 \big)} \Bigg) \nonumber \\
 \overset{(d)}{=} & h_j  + \frac{1}{2} \sum_{k \in {\mathcal N}_{j}} \log\Big(\frac{1 + e^{2\beta + 2\Gamma_{k}^{t-1}} }{e^{2\beta} + e^{2\Gamma_{k}^{t-1}}}\Big)  \nonumber
\end{align}
\begin{itemize}
\item $(a)$ holds because conditioned on $\tau_{j}$, $(N_{j},\tilde{\tau}_{j})$ are independent of the rest of the tree, and $(T_{k}^{t-1},\tilde{\tau}_{T_{k}^{t-1}})$ are independent and identically distributed random variables $\forall k \in N_{j}$,  
\item $(b)$ holds also because conditioned on $\tau_{j}$, $N_{j}$ and $\tilde{\tau}_{j}$ are independent, 
\item $(c)$ holds because $N_{j} \sim \text{Pois}(\frac{a+b}{2})$ $\forall j \in T^{t}$, and for a node $j$, $N_{j}$ children are generated $\sim \text{Pois}(\frac{a+b}{2})$, then for each node $k \in N_{j}$, $\tau_{k} = \tau_{j}$ with probability $\frac{a}{a+b}$ and $\tau_{k} = -\tau_{j}$ with probability $\frac{b}{a+b}$, 
\item $(d)$ holds from the definition of $\beta$.
\end{itemize}
\end{proof}
The above result clarifies the connection between inference on trees and the community detection problem addressed in this paper. Let $G_{i}^{t}$ be the sub-graph of $G$ induced by the nodes whose distance to $i$ is at most $t$, and $\boldsymbol{x}_{A}$ be a vector consisting of labels of nodes in a set of nodes A. Then, the following Lemma, proved in~\cite{Main_paper2}, shows the feasibility of approximating $(G_{i}^{t},\boldsymbol{x}_{G_{i}^{t}},\boldsymbol{y}_{G_{i}^{t}})$ by $(T_{i}^{t},\tau_{T_{i}^{t}},\tilde{\tau}_{T_{i}^{t}})$ with probability approaching one under certain conditions on the depth $t$.

\begin{Lemma}[\cite{Main_paper2}]
\label{Le.2}
For $ t= t(n)$ such that $(\frac{a+b}{2})^{t} = n^{o(1)}$, there exists a coupling between $(G,\boldsymbol{x},\boldsymbol{y})$ and $(T,\tau,\tilde{\tau})$ such that $(G_{i}^{t},\boldsymbol{x}_{G_{i}^{t}},\boldsymbol{y}_{G_{i}^{t}}) = (T_{i}^{t},\tau_{T_{i}^{t}},\tilde{\tau}_{T_{i}^{t}})$ with probability converging to $1$.
\end{Lemma}

Lemma~\ref{Le.2} suggests that the tree-based log likelihood ratio $\Gamma_{i}^{t}$, calculated in Lemma~\ref{Le.1},  is an asymptotically accurate representation for belief propagation in our problem. Let $\hat{\boldsymbol{x}}_{BP^{t}}$ be the output of the belief propagation algorithm after t iterations. The details of the belief propagation algorithm is presented in Table~\ref{tab}.

\begin{table}
\caption{Belief propagation algorithm with side information.}
\label{tab}
\centering
\begin{tabular}{|p{3in}|}
\hline
Belief Propagation Algorithm \\
\hline
1: Input: $n, t\in\mathbb{N}$, $G$, $\boldsymbol{y}$.\\

2: Initialize: Set $R^{0}_{i\to j} = 0$, $\forall i \in G$ and $j \in \mathcal{N}_{i}$.\\

3: For all $i \in G$ and $j \in \mathcal{N}_{i}$, run for $t-1$ iterations:\\

$R^{t-1}_{i\to j} = h_i + \sum_{k \in \mathcal{N}_{i} \backslash\{j\}} \log\Big(\frac{1 + e^{2\beta + 2R^{t-2}_{k\to i}} }{e^{2\beta} + e^{2R^{t-2}_{k\to i}}}\Big)$
\\

4: For all $i \in G$, compute:\\

$R^{t}_{i} = h_i + \sum_{k \in \mathcal{N}_{i}} \log\Big(\frac{1 + e^{2\beta + 2R^{t-1}_{k\to i}} }{e^{2\beta} + e^{2R^{t-1}_{k\to i}}}\Big)$\\

5: Return $\hat{\boldsymbol{x}}_{BP^{t}}$ with $\hat{\boldsymbol{x}}_{BP^{t}}(i) = 2 \times 1_{\{R^{t}_{i} \geq 0\}} -1$.\\
\hline

\end{tabular}

\end{table}

Define 
\[
p_{G,\boldsymbol{y}}(\hat{\boldsymbol{x}}) \triangleq \frac{1}{n} \sum_{i=1}^{n} \mathop{\mathbb{P}}\{x_{i} \neq \hat{x_{i}}\}
\]
to be the expected fraction of misclassified nodes by an estimator $\hat{\boldsymbol{x}}$. The following lemma characterizes the asymptotic average behavior of grah-wide error as characterized by $p_{G,\boldsymbol{y}}(\hat{\boldsymbol{x}}_{BP^{t}})$.

\begin{Lemma}\label{Le.3}
For $ t= t(n)$ such that $(\frac{a+b}{2})^{t} = n^{o(1)}$, $\lim_{n\to\infty} |p_{G,\boldsymbol{y}}(\hat{\boldsymbol{x}}_{BP^{t}}) - q^{*}_{T^{t}}| = 0$.
\end{Lemma}

\begin{proof}
By Lemma~\ref{Le.2},  $(G_{i}^{t},\boldsymbol{x}_{G_{i}^{t}},\boldsymbol{y}_{G_{i}^{t}}) = (T_{i}^{t},\tau_{T_{i}^{t}},\tilde{\tau}_{T_{i}^{t}})$ with probability converging to 1. This implies that $R^{t}_{i} = \Gamma_{i}^{t}$, and hence, $p_{G,\boldsymbol{y}}(\hat{\boldsymbol{x}}_{BP^{t}}) = q^{*}_{T^{t}} + o(1)$, where the $o(1)$ term comes from the coupling error of Lemma~\ref{Le.2}.
\end{proof}

So far the results hold for all $a$ and $b$ as long as $\frac{(a+b)}{2} = n^{o(1)}$ and $\frac{a}{b} = \Theta(1)$. Now let $a = b + \mu \sqrt{b}$, for a fixed positive constant $\mu$. Let $U_{+}$ and $U_{-}$ be two random variables drawn according to the distribution of $h_i$ conditioned respectively on $\tau_i =1$ and $\tau_i =-1$. Then the following theorem describes a density evolution that evaluates $q^{*}_{T^{t}}$.

\begin{Theorem}\label{Th.1}
Assume as $n\to\infty$, $b\to\infty$ and $\frac{a-b}{\sqrt{b}} \to \mu$, for a fixed positive constant $\mu$. Also, let $h(\nu) = \mathbb{E}[\tanh(\nu + \sqrt{\nu}Z + U_{+})]$, where $Z\sim \mathcal{N}(0,1)$. Define $\bar{\nu}$ to be the smallest fixed point of $\nu = \frac{\mu^{2}}{4} h(\nu)$. Then:

\begin{align}
\label{eq.err.1}
\lim_{t\to\infty}\lim_{n\to\infty} p_{G,\boldsymbol{y}}(\hat{\boldsymbol{x}}_{BP^{t}}) =&  \frac{1}{2} \Big(\mathbb{E}_{U_{+}}\big[ Q\big(\frac{\bar{\nu} + U_{+}}{\sqrt{\bar{\nu}}}\big) \big] \twocolbreak + \mathbb{E}_{U_{-}}\big[ Q\big(\frac{\bar{\nu} - U_{-}}{\sqrt{\bar{\nu}}}\big) \big]  \Big)
\end{align}
\end{Theorem}

\begin{proof}
The proof has similarities with~\cite{Main_paper2}. For brevity, we only describe the new developments compared with~\cite{Main_paper2} and  the corresponding arguments.

Define 
\begin{equation}
F(x) \triangleq \frac{1}{2} \log\big( \frac{e^{2x+2\beta}+1}{e^{2x}+e^{2\beta}} \big)
\label{eq:defF}
\end{equation}
and
for all $t \geq 1$, $\Phi_{j}^{t} = \sum_{k \in N_{j}} F(\Phi_{k}^{t-1} +  h_{k})$. Thus, for all $t \geq 0$,
\begin{align}
\Gamma_{j}^{t} &=  h_{j} + \Phi_{j}^{t}.
\end{align}

We are interested in the moments of $\Phi_{j}^{t}$ conditioned on node label $\tau_j =-1$ and $\tau_j =1$. For convenience of notation, we define new random variables $W_{+}^{t}$ and $W_{-}^{t}$ whose distribution is identical to $\Phi_{j}^{t}$  when $\tau_{j}$ is equal to $1$ and $-1$, respectively.
\begin{Lemma}\label{Le.4}
For all $t \geq 0$,
\begin{align}\nonumber
& \mathbb{E}[W_{\pm}^{t+1}] = \pm \frac{\mu^{2}}{4}\mathbb{E}[\tanh(W_{+}^{t}+U_{+})] + O(a^{-\frac{1}{2}}) \\ \nonumber
& \text{var}(W_{\pm}^{t+1}) = \frac{\mu^{2}}{4}\mathbb{E}[\tanh(W_{+}^{t}+U_{+})] + O(a^{-\frac{1}{2}}) \nonumber
\end{align}
\end{Lemma}
\begin{proof}

The proof for $\mathbb{E}[W_{-}^{t+1}]$ and $\text{var}(W_{-}^{t+1})$ departs from~\cite[Lemma 7.1]{Main_paper2} in the distribution of $U_{\pm}$.

Define $\psi(x) = \log(1+x) - x$. It then follows from Taylor expansion that $|\psi(x)| \leq x^{2}$. Then, $F(x)$, defined in~\eqref{eq:defF}, can be written as:
\begin{align}
F(x) &= -\beta + \frac{1}{2} \log\big( 1 + \frac{e^{4\beta}-1}{e^{-2(x-\beta)}+1} \big)\nonumber\\ &  = -\beta + \frac{e^{4\beta}-1}{2}f(x) + \frac{1}{2}\psi\big( (e^{4\beta}-1)f(x) \big)\nonumber
\end{align}
where $f(x) = \frac{1}{1+e^{-2(x-\beta)}}$. It then follows that:
\begin{align}
\Phi_{j}^{t+1} =& \sum_{k \in N_{j}} F(\Phi_{k}^{t} +  h_{k}) \nonumber\\
=& \sum_{k \in N_{j}} \Big[ -\beta + \frac{e^{4\beta}-1}{2}f(\Phi_{k}^{t} +  h_{k})  \twocolbreak + \frac{1}{2} \psi\big( (e^{4\beta}-1)f(\Phi_{k}^{t} +  h_{k})\big) \Big]
\end{align}
Calculating the mean of the two sides of equation above conditioned on $\tau_j=\pm1$,
\begin{align}
&  \mathbb{E}[W_{+}^{t+1}] - \mathbb{E}[W_{-}^{t+1}] \twocolbreak = (\frac{e^{4\beta}-1}{4})(a-b)\mathbb{E}\Big[f(W_{+}^{t} + U_{+}) -  f(W_{-}^{t} + U_{-})\Big] \nonumber\\ 
  & + \frac{a-b}{4}\mathbb{E}\Big[ \psi\big( (e^{4\beta}-1)f(W_{+}^{t} + U_{+}) \big) \twocolbreak- \psi\big( (e^{4\beta}-1)f(W_{-}^{t} + U_{-}) \big) \Big] \label{bin.symm.new.1}
\end{align}

By the definition of $\Gamma_{j}^{t}$ and a change of measure, it follows that $\mathbb{E}[g(\Gamma_{j}^{t}) | \tau_{j} = -1] = \mathbb{E}[g(\Gamma_{j}^{t}) e^{-2\Gamma_{j}^{t}} | \tau_{j} = 1]$ for any measurable function $g$ such that the expectations are well defined. Also, notice that:
 \begin{align}
(\frac{e^{4\beta}-1}{4})(a-b) & = \frac{\mu^{2}}{2} \frac{a+b}{2b} = \frac{\mu^{2}}{2}(1 + \frac{a-b}{2b}) \\ 
& = \frac{\mu^{2}}{2} + O(a^{-\frac{1}{2}}) \label{bin.symm.new.2}
\end{align}

 Moreover, since $|\psi(x)| \leq x^{2}$ and $|f(x)| \leq 1$, it follows that $\psi\big( (e^{4\beta}-1)f(W_{+}^{t} + U_{+}) \big) - \psi\big( (e^{4\beta}-1)f(W_{-}^{t} + U_{-}) \big) \leq 2(e^{4\beta}-1)^{2}$. Therefore,
 \begin{align}
   \frac{a-b}{4}\mathbb{E}\Big[ \psi\big( (e^{4\beta}-1)f(W_{+}^{t} + U_{+}) \big) \twocolnewline- \psi\big( (e^{4\beta}-1)f(W_{-}^{t} + U_{-}) \big) \Big]  & \leq \frac{a-b}{2}(e^{4\beta}-1)^{2}\nonumber \\ 
   & = O(a^{\frac{-1}{2}}) \label{bin.symm.new.3}
\end{align}

 Combining~\eqref{bin.symm.new.1},~\eqref{bin.symm.new.2}, and~\eqref{bin.symm.new.3}, 
 \begin{align}
\mathbb{E}[W_{+}^{t+1}] &= \mathbb{E}[W_{-}^{t+1}] \twocolbreak+ \big( \frac{\mu^{2}}{2} + O(a^{\frac{-1}{2}}) \big)\mathbb{E}\big[ f(W_{+}^{t} + U_{+}) (1 - e^{-2(W_{+}^{t} + U_{+})})\big] \twocolbreak + O(a^{\frac{-1}{2}}) \nonumber \\ 
& \overset{(a)}{=} \frac{\mu^{2}}{4}\mathbb{E}[\tanh(W_{+}^{t}+U_{+})] - O(a^{\frac{-1}{2}})\mathbb{E}[e^{-2(W_{+}^{t} + U_{+})}]  \twocolbreak+ O(a^{\frac{-1}{2}}) \nonumber \\ 
&\overset{(b)}{=} \frac{\mu^{2}}{4}\mathbb{E}[\tanh(W_{+}^{t}+U_{+})] + O(a^{\frac{-1}{2}})
\end{align}
where $(a)$ holds from the definition of $f(x)$, the definition of $\tanh(x)$ and the fact that $f(x) = \frac{1}{1 + e^{-2x}} + O(a^{\frac{-1}{2}})$ and $(b)$ holds because by change of measure $\mathbb{E}[e^{-2(W_{+}^{t} + U_{+})}] = \mathbb{E}[e^{-2(W_{-}^{t} + U_{-})}e^{2(W_{-}^{t} + U_{-})}] = 1$. This concludes the proof for $\mathbb{E}[W_{+}^{t+1}]$. The proof for $\text{var}(W_{+}^{t+1})$ follows similarly.
\end{proof}

\begin{Lemma}\label{Le.5}
Assume $\alpha_{-,m}, \alpha_{+,m}$ are constants as $n \to \infty$. Let $h(\nu) = \mathbb{E}[\tanh(\nu + \sqrt{\nu}Z + U_{+})]$, where $Z\sim \mathcal{N}(0,1)$. Define $(\nu_{t}: t\geq0)$ recursively by $\nu_{0} =0$ and $\nu_{t+1} = \frac{\mu^{2}}{4}h(\nu_{t})$. Then, for any fixed $t \geq 0$, as $n \to\infty$:
\begin{equation}
\sup_{x} \left| \mathbb{P}\bigg\{ \frac{W_{\pm}^{t} \mp \nu_{t}}{\sqrt{\nu^{t}}} \leq x \bigg\} - \mathbb{P}\{ Z \leq x \} \right| = O(a^{-\frac{1}{2}})
\end{equation}
\end{Lemma}

The proof of Lemma~\ref{Le.5} departs from~\cite[Lemma 7.3]{Main_paper2} only in the distribution of $U_{\pm}$, and is therefore omitted for brevity.

In view of Lemmas~\ref{Le.4},~\ref{Le.5}, for all $j$, $(\Phi_{j}^{t} | \tau_{j} = \pm 1) \sim \mathcal{N}(\pm \nu_{t}, \nu_{t})$. Hence, 
\begin{align}\nonumber
\lim_{n\to \infty} \mathbb{P}(\Gamma_{j}^{t} > 0 | \tau_{j} = -1) 
& = \mathbb{E}_{U_{-}}\big[ Q\big(\frac{\bar{\nu} - U_{-}}{\sqrt{\bar{\nu}}}\big) \big]
\end{align}
\begin{align}\nonumber
\lim_{n\to \infty} \mathbb{P}(\Gamma_{j}^{t} < 0 | \tau_{j} = 1)
& = \mathbb{E}_{U_{+}}\big[ Q\big(\frac{\nu_{t} + U_{+}}{\sqrt{\nu_{t}}}\big) \big]
\end{align}
where $Q(x) = \int_{x}^{\infty} \frac{1}{\sqrt{2\pi}} e^{\frac{-y^{2}}{2}} dy$. Hence, from Lemma~\ref{Le.3}, 

\begin{align*}
\label{eq.6}
\lim_{n\to\infty} &p_{G,\boldsymbol{y}}(\hat{\boldsymbol{x}}_{BP^{t}}) =  \lim_{n\to\infty} q^{*}_{T^{t}} \twocolbreak= \frac{1}{2} \Big(\mathbb{E}_{U_{+}}\big[ Q\big(\frac{\nu_{t} + U_{+}}{\sqrt{\nu_{t}}}\big) \big] + \mathbb{E}_{U_{-}}\big[ Q\big(\frac{\nu_{t} - U_{-}}{\sqrt{\nu_{t}}}\big) \big]  \Big) 
\end{align*}

It remains to show that $\lim_{t \to\infty} \nu_{t} = \bar{\nu}$. 

\begin{Lemma}\label{Le.6}
Let $h(\nu) = \mathbb{E}[\tanh(\nu + \sqrt{\nu}Z + U_{+})]$, where $Z\sim \mathcal{N}(0,1)$. Then, $h(\nu)$ is continuous  on $[0,\infty]$ and $ h^{'}(\nu) \geq 0$ for $\nu \in (0,\infty)$.
\end{Lemma}

The proof of Lemma~\ref{Le.6} departs from~\cite[Lemma 7.4]{Main_paper2} only in the distribution of $U_{\pm}$, and is therefore omitted for brevity.

Recall that $\nu_{0} = 0$. By direct substitution $\nu_{0} \leq \nu_{1}$. Now, let $\nu_{t+1} \geq \nu_{t}$. By Lemma~\ref{Le.6}, 
\begin{equation}
\nu_{t+2} - \nu_{t+1} = \frac{\mu^{2}}{4} (h(\nu_{t+1}) - h(\nu_{t})) = \frac{\mu^{2}}{4}h^{'}(x)
\end{equation}
for some $x \in (\nu_{t},\nu_{t+1})$. By Lemma~\ref{Le.6}, $h^{'}(x) \geq 0$ for $x \in (0,\infty)$. Thus, $\nu_{t+2} \geq \nu_{t+1}$, and hence, it has been shown by induction on $t$ that $\nu_{t}$ is non-decreasing in $t$. Also, note that $\nu_{0} = 0 \leq \bar{\nu}$. If we assume that $\nu_{t} \leq \bar{\nu}$, then by monotonicity of $h$, we have: $\nu_{t+1} = \frac{\mu^{2}}{4} h(\nu_{t}) \leq \frac{\mu^{2}}{4} h(\bar{\nu}) = \bar{\nu}$. Thus, $\lim_{t \to\infty} \nu_{t} = \bar{\nu}$.

\end{proof}


\subsection{Exit Analysis}\label{EX.}

Equation~\eqref{eq.err.1} characterizes the asymptotic residual error of belief propagation for recovering binary symmetric communities with side information. However, we seek answers to some natural and interesting questions that are not directly apparent by inspection from~\eqref{eq.err.1}, such as: What is the effect of quality and quantity of side information on the residual error? How is this related to the amount of information provided by the graph about node labels? Can side information dominate the performance of belief propagation for community detection, and if so, under what conditions does that happen? In this section, we show that EXIT charts can provide answers to these questions, via existence and location of crossing points of EXIT curves.

We begin by calculating the mutual information between the label of node $i$, $x_{i}$, and its belief at time $t$, namely $R_{i}^{t}$.
\begin{align}
I(x_{i},R_{i}^{t}) \nonumber \\
 = &1-H(x_{i}|R_{i}^{t})  \nonumber \\
=& 1- \frac{1}{2} \int_{-\infty}^{\infty} \Bigg(\sum_{m=1}^{M} \alpha_{+,m} \frac{e^{\frac{-(y-(v_{t}+ h_{m}))^{2}}{2v_{t}}}}{\sqrt{2\pi v_{t}}}\Bigg) \twocolbreak{} \log_{2}\bigg(1 + \frac{\sum_{m=1}^{M} \alpha_{-,m} e^{\frac{-(y-(-v_{t}+ h_{m}))^{2}}{2v_{t}}}}{\sum_{m=1}^{M} \alpha_{+,m} e^{\frac{-(y-(v_{t}+ h_{m}))^{2}}{2v_{t}}}} \bigg) dy  \nonumber \\
&-\frac{1}{2} \int_{-\infty}^{\infty} \Bigg(\sum_{m=1}^{M} \alpha_{-,m} \frac{e^{\frac{-(y-(-v_{t}+ h_{m}))^{2}}{2v_{t}}}}{\sqrt{2\pi v_{t}}}\Bigg) \twocolbreak{} \log_{2}\Bigg(1 + \frac{ \sum_{m=1}^{M} \alpha_{+,m} e^{\frac{-(y-(v_{t}+ h_{m}))^{2}}{2v_{t}}}}{\sum_{m=1}^{M} \alpha_{-,m} e^{\frac{-(y-(-v_{t}+ h_{m}))^{2}}{2v_{t}}}} \Bigg) dy \label{mutual}
\end{align}

For simplicity and to show the power of EXIT analysis in drawing insights that cannot be easily deduced from belief propagation equations, we consider a concrete example with $M=3$. More precisely, for each node $i$, we observe $y_{i} = x_{i}$ with probability $\epsilon(1 - \alpha)$ or $y_{i} = -x_{i}$ with probability $\epsilon\alpha$ or $y_{i} = 0$ with probability $1-\epsilon$, independently at random, for $\alpha \in (0,0.5)$ and $\epsilon \in [0,1]$. Thus, $U_{+} = -U_{-}$, where $U_{+} \in \{ \gamma, -\gamma, 0 \}$ with probabilities $\epsilon(1 - \alpha)$, $\epsilon\alpha$ and $1-\epsilon$, respectively, where $\gamma \triangleq \frac{1}{2} \log\big( \frac{1-\alpha}{\alpha}\big)$. Note that for fixed $\alpha$ and $\epsilon$, $I(x_{i},R_{i}^{t})$ is function of $\nu_{t}$ only. Hence, we will denote it by $J(\nu_{t})$.

Based on the belief propagation algorithm described in Table~\ref{tab}, at iteration $t$, node $i$ receives the beliefs of all nodes $j \in N(i)$ calculated at iteration $(t-1)$. We denote the information node $i$ receives from node $j$ as $I_{in}$. Then, node $i$ computes the new information it has at iteration $t$. We denote this information as $I_{out}$. Both $I_{in}$ and $I_{out}$ can be calculated using~\eqref{mutual} as $J(\nu_{t-1})$ and $J(\nu_{t})$, respectively. Since $J(\nu_{t})$ is monotonically increasing in $\nu_{t}$~\cite{exit}, $J(\nu_{t})$ is reversible. Thus, $\nu_{t} = J^{-1}(I(x_{i},R_{i}^{t}))$.  Moreover, $\nu_{t-1}$ and $\nu_{t}$ are related by
\[
\nu_{t+1} = \frac{\mu^{2}}{4}h(\nu_{t})
\]
therefore  $I_{in}$ and $I_{out}$ for node $i$ are related as follows
\begin{align}
 & I_{out} = \nonumber \\
 &  J\Bigg(\frac{\mu^{2}}{4} \bigg[ \epsilon(1-\alpha) \mathbb{E}_{Z}[\tanh(J^{-1}(I_{in}) + \sqrt{J^{-1}(I_{in})}Z + \gamma)] + \nonumber \\
& \epsilon\alpha \mathbb{E}_{Z}[\tanh(J^{-1}(I_{in}) + \sqrt{J^{-1}(I_{in})}Z - \gamma)] + \twocolbreak{} (1-\epsilon) \mathbb{E}_{Z}[\tanh(J^{-1}(I_{in}) + \sqrt{J^{-1}(I_{in})}Z )] \bigg] \Bigg) 
\end{align}

There is a fundamental difference between using EXIT charts in the context of community detection in stochastic block models and EXIT charts in the standard context of coding theory. Taking Low Density Parity Check (LDPC) Codes as an example, each variable node $i$ receives from a check node $j$ the information or belief of that check node about whether the variable node is one or zero. Thus, the input log-likelihood ratio received by variable node $i$ is actually calculated conditioned on the value of the variable node $i$. Community detection presents a different scenario: Each node $i$ receives the belief of node $j$. However, the belief of node $j$ is calculated conditioned on the value of node $j$, not the value of node $i$. This reflects the fundamental differences between the bipartite graph representing FEC codewords and a random graph representing relationships of randomly distributed node labels. The former is fundamentally asymmetric, where parity nodes carry no new information conditioned on bit nodes. On the contrary, in community detection, nodes are (stochastically) symmetric and all of them carry information.

In community detection, for a node $i$ at iteration t, we define $I_{in} = I(x_{j},R_{j}^{t-1})$, and $I_{out} = I(x_{i},R_{i}^{t})$. In other words, the amount of information transferred {\em from} each node outward represents how confident (in terms of mutual information) is the belief of that node about the value of its own label. 

To compute $J$ and $J^{-1}$, we apply curve fitting using the Levenberg Marquardt algorithm~\cite{exit}.
Figures~\ref{fig1},~\ref{fig2}, and~\ref{fig3} show the EXIT curves for different values of $\mu$, $\alpha$ and $\epsilon$. From these figures, we can deduce the following:

\begin{figure}
\centering
\includegraphics[width=\Figwidth]{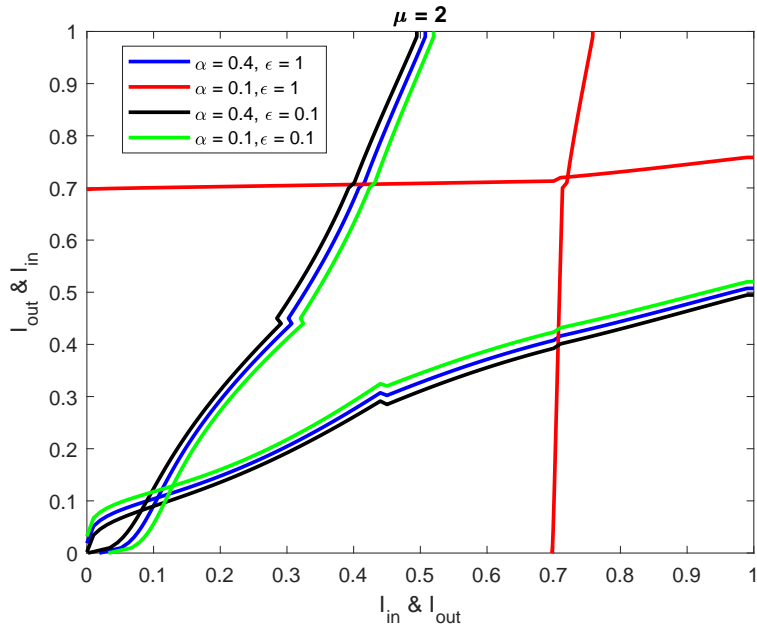}
\caption{EXIT Chart for $\mu =2$.}
\label{fig1}
\end{figure}
\begin{figure}
\centering
\includegraphics[width=\Figwidth]{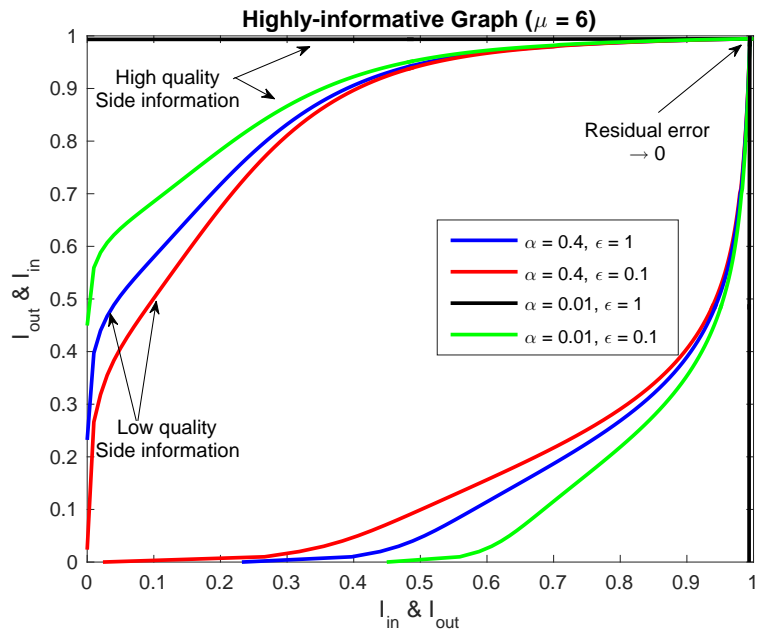}
\caption{EXIT Chart for $\mu =6$.}
\label{fig2}
\end{figure}
\begin{figure}
\centering
\includegraphics[width=\Figwidth]{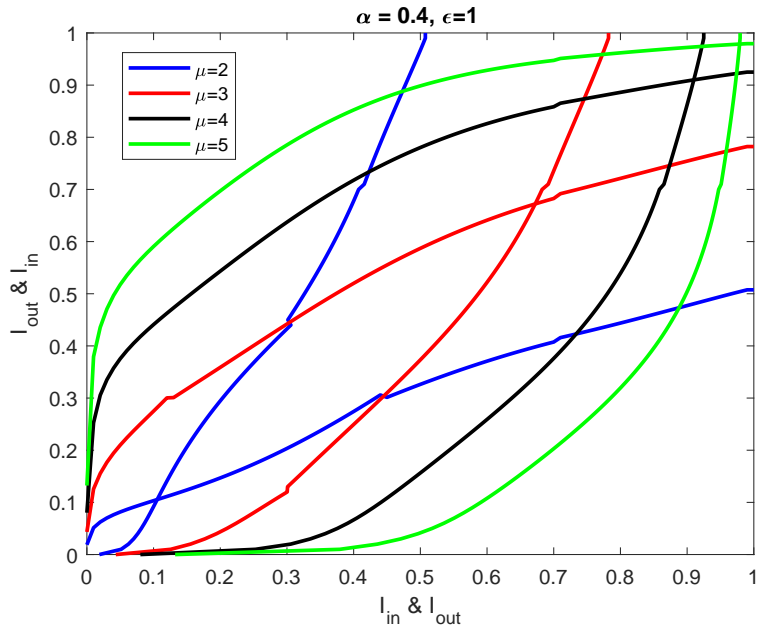}
\caption{EXIT Chart for $\alpha =0.4$ and $\epsilon=1$.}
\label{fig3}
\end{figure}

\begin{itemize}
\item Side information, with any quantity (any $\epsilon \neq 0$), regardless of the quality (e.g. $\alpha =0.4$), breaks the symmetry. Note that without side information the curves get stuck at the trivial $(0,0)$ point, implying that the belief propagation algorithm is a trivial random guessing estimator~\cite{Saad:ISIT2016}. This is true for all values of $\mu$.

\item The starting point of the curves, which indicates the quality of the initial estimate, depends crucially on the values of $\mu, \alpha, \epsilon$. For small values of $\mu$, e.g. $\mu =2$, EXIT charts reveal that the quantity of side information is not very important unless its quality is excellent. This can be seen in Figure~\ref{fig1}: when $\alpha =0.4$, the starting point for all values of $\epsilon \neq 0$ is almost the same. On the other hand, when $\alpha = 0.1$, the effect of $\epsilon$ on the starting point of the curve can be very significant, and the gap is around $0.7$ between $\epsilon = 1$ and $\epsilon =0.1$. For large values of $\mu$, e.g. $\mu =6$, the behavior changes. EXIT charts show that the effect of $\epsilon$ becomes more significant even when $\alpha =0.4$. This is because larger values of $\mu$ imply larger difference between $a$ and $b$, which means easier detection (quick convergence). Therefore  when $\mu$ is large, the quality of the initial guess can make a bigger difference, proportionally.

\item The intersection points on the curve exhibit almost the same behavior as the starting point. Note that the intersection point determines the value of $\bar{\nu}$, which determines the probability of error. In Figure~\ref{fig1}, when $\alpha =0.4$,  the intersection points are very close in value for $\epsilon \neq 0$. This shows that the quantity of side information does not enhance the performance of belief propagation for small values of $\mu$. On the other hand, when $\alpha = 0.1$,  the effect of $\epsilon$ on the intersection point of the curve (i.e., probability of error) is significant, even when $\mu =2$.

\item EXIT charts also show that when the graph is not very informative, e.g., $\mu=2$, even when side information provides significant information, e.g., when $\alpha=0.1, \epsilon=1$, the residual error does not improve markedly over the course of iterations. On the other hand, for highly-informative graphs, e.g., $\mu =6$, even when side information provides a small amount of information, e.g., when $\alpha=0.4, \epsilon=0.1$, the eventual residual error improves significantly compared with the starting point. 

 \item Although side information can break symmetry, even with high quality, e.g., $\alpha =0.1$, unless $\epsilon \to 1$, one cannot hope to reach a vanishing fraction of misclassified nodes for a graph with small $\mu$. This stems from the fact that the two communities are symmetric and for nodes with erased side information, the only source of information is the messages coming from its neighbors.

\item When $\mu = 6$, for all values of $\alpha \in (0,0.5)$ and $\epsilon \in (0,1]$, one may achieve a vanishing fraction of misclassified nodes. This is because the only intersection point on the curve is approaching $(1,1)$, which is the maximum mutual information available for binary variables.

\item Figure~\ref{fig3} shows that as $\mu$ increases, there is always an intersection point. This suggests that one could not hope for vanishing residual error, i.e., weak recovery, except when $\mu \to \infty$ or $\alpha \to 0$. This suggests that belief propagation for recovering binary symmetric communities with side information does not have a phase transition for a finite $\mu$.
\end{itemize}


\section{One Community Stochastic Block Model}\label{Bel.1}

We begin by studying the performance of belief propagation on a random tree with side information. Then, we show that the same performance is possible on a random {\em graph} drawn according to the one community stochastic block model with side information, using a coupling lemma~\cite{recovering_one}.


Let $T$ be an infinite tree with nodes indexed by variable $i$, each of them possessing a label $\tau_{i} \in \{0,1\}$.  The root is node $i=0$. The sub-tree of depth $t$ rooted at node $i$ is denoted $T_{i}^{t}$. The sub-tree rooted at $i=0$ with depth $t$ is referenced often and is denoted simply $\CroppedTree$. Unlike the random graph counterpart, the tree and its node labels are generated together as follows: $\tau_{0}$ is a Bernoulli-$\frac{K}{n}$ random variable. For any $i\in T$, the number of its children with label $1$ is a random variable $H_i$ that is Poisson with parameter $Kp$ if $\tau_i=1$, and Poisson with parameter $Kq$ if $\tau_i=0$. The number of children of node $i$ with label $0$ is a random variable $F_i$ which is Poisson with parameter $(n-K)q$, regardless of the label of node $i$. The side information $\tilde{\tau}_i$ takes value in a finite alphabet $\{u_{1},\cdots,u_{M}\}$. The set of all labels in $T$ is denoted with $\TreeLabels$, all side information with $\TreeSideInfo$, and the labels and side information of $\CroppedTree$ with $\TreeLabels^t$ and $\TreeSideInfo^t$ respectively.  The likelihood of side information continues to be denoted by $\alpha_{+,m},\alpha_{-,m}$, as earlier.

The goal is to infer the label $\RootLabel$ given observations $\CroppedTree$ and $\TreeSideInfo^t$. The error probability of the estimator $\RootEstimator(\CroppedTree,\TreeSideInfo^t)$ is:
\begin{align}
\label{error.general}
p_{e}^{t} &\triangleq \frac{K}{n} \mathbb{P}(\RootEstimator = 0 | \RootLabel = 1) + \frac{n-K}{n} \mathbb{P}(\RootEstimator = 1 | \RootLabel = 0)
\end{align}
The maximum a posteriori (MAP) detector minimizes $p_{e}^{t}$ is given by $\hat{\tau}_{MAP} = 1_{\{\LLRCroppedTree \geq \nu\}}$, where $\LLRCroppedTree$ is the log likelihood ratio, 
\begin{align}
\label{likeli}
\LLRCroppedTree \triangleq \log\bigg(\frac{\mathbb{P}(\CroppedTree,\TreeSideInfo^t | \RootLabel = 1)}{\mathbb{P}(\CroppedTree,\TreeSideInfo^t | \RootLabel = 0)}\bigg)
\end{align}
and $\nu = \log(\frac{n-K}{K})$.
\begin{Lemma}
\label{recursive}
Let ${\mathcal N}_{i}$ denote the children of node $i$, $N_i\triangleq|{\mathcal N}_i|$ and $h_{i} \triangleq \log\big(\frac{\mathbb{P}(\tilde{\tau}_{i} | \tau_i =1)}{\mathbb{P}(\tilde{\tau}_{i} | \tau_i =0)}\big)$. Then,
\begin{align}
\Gamma_{i}^{t+1} &= -K(p-q) + h_{i} + \sum_{k \in {\mathcal N}_{i}} \log\bigg(\frac{\frac{p}{q}e^{\Gamma_{k}^{t} - \nu}+1}{e^{\Gamma_{k}^{t} - \nu}+1}\bigg)
\end{align}
\end{Lemma}

\begin{proof}
The independent splitting property of the Poisson distribution is used to give an equivalent description of the numbers of children having a given label for any vertex in the tree, as follows. The set of children of node $i$ is denoted ${\mathcal N}_i$ with cardinality $N_i=|{\mathcal N}_i|$. If $\tau_{i} = 1$, the number of its children $N_{i} \sim \text{Poi}(Kp + (n-K)q)$ and each of these children $j$, independently of everything else has label $\tau_{j} =1$ with probability $\frac{Kp}{Kp + (n-K)q}$  and  $\tau_{j} =0$ with probability $\frac{(n-K)q}{Kp + (n-K)q}$. If $\tau_{i} = 0$ the number of its children $N_{i} \sim \text{Poi}(nq)$ and each of these children $j$, independent from everything else, has label $\tau_{j} =1$ with probability $\frac{K}{n}$ and $\tau_{j} =0$ with probability $\frac{(n-K)}{n}$. Finally, for each node $i$ in the tree, side information $\tilde{\tau}_{i}$ is observed according to $\alpha_{+,m}, \alpha_{-,m}$. Then:
\begin{align} 
\Gamma_{0}^{t+1} &= \log\Bigg( \frac{\mathbb{P}(T^{t+1},\tilde{\tau}^{t+1} | \RootLabel = 1)}{\mathbb{P}(T^{t+1},\tilde{\tau}^{t+1} | \RootLabel = 0)} \Bigg) \nonumber \\ 
& = \log\Bigg( \frac{\mathbb{P}(N_{0} , \tilde{\tau}_{0}, \{T_{k}^{t}\}_{k \in {\mathcal N}_0},\{\tilde{\tau}_{k}^{t}\}_{k \in {\mathcal N}_0} | \RootLabel = 1)}{\mathbb{P}(N_{0} , \tilde{\tau}_{0}, \{T_{k}^{t}\}_{k \in {\mathcal N}_0},\{\tilde{\tau}_{k}^{t}\}_{k \in {\mathcal N}_0} | \RootLabel = 0)} \Bigg) \nonumber\\ 
& \overset{(a)}{=} \log\Bigg( \frac{\mathbb{P}\big( N_{0},\tilde{\tau}_{0} | \RootLabel =1\big)}{\mathbb{P}\big( N_{0},\tilde{\tau}_{0} | \RootLabel = 0\big)} \Bigg) \twocolbreak \includeonetwocol{}{\hspace{0.2in}} +  \log\Bigg( \frac{\prod_{k \in {\mathcal N}_0} \mathbb{P}\big(T_{k}^{t},\tilde{\tau}_{k}^{t} | \RootLabel = 1\big)}{\prod_{k \in {\mathcal N}_0} \mathbb{P}\big(T_{k}^{t},\tilde{\tau}_{k}^{t} | \RootLabel = 0\big)} \Bigg) \nonumber \\ 
& \overset{(b)}{=} \log\Bigg( \frac{\mathbb{P}\big( N_{0} | \RootLabel =1\big)}{\mathbb{P}\big( N_{0} | \RootLabel =0\big)} \Bigg) + \log\Bigg( \frac{\mathbb{P}\big( \tilde{\tau}_{0} | \RootLabel =1\big)}{\mathbb{P}\big( \tilde{\tau}_{0} | \RootLabel =0\big)} \Bigg) \nonumber \\
& \includeonetwocol{}{\hspace{0.1in}} + \sum_{k \in {\mathcal N}_0} \log\Bigg( \frac{ \sum_{\tau_{k} \in \{0,1\} } \mathbb{P}\big( T_{k}^{t},\tilde{\tau}_{k}^{t} | \tau_{k} \big) \mathbb{P}\big( \tau_{k}|\RootLabel=1 \big) }{\sum_{\tau_{k} \in \{0,1\} }  \mathbb{P}\big( T_{k}^{t},\tilde{\tau}_{k}^{t} | \tau_{k} \big) \mathbb{P}\big( \tau_{k}|\RootLabel=0 \big)} \Bigg) \nonumber  \\ 
& \overset{(c)}{=} -K(p-q) + h_{0} + \sum_{k \in {\mathcal N}_0} \log(\frac{\frac{p}{q}e^{\Gamma_{k}^{t} - \nu}+1}{e^{\Gamma_{k}^{t} - \nu}+1})\label{rec.eq}
\end{align} 
where 
\begin{itemize}
\item $(a)$ holds because conditioned on $\RootLabel$ $(N_{0},\tilde{\tau}_{0})$ are independent of the rest of the tree and also $(T_{k}^{t},\tilde{\tau}_{k}^{t})$ are independent random variables $\forall k \in {\mathcal N}_0$,  
\item $(b)$ holds because conditioned on $\RootLabel$, $N_{0}$ and $\tilde{\tau}_{0}$ are independent, 
\item $(c)$ holds by the definition of $N_{0}$ and $h_{0}$ and because $\tau_{k}$ is Bernoulli-$\frac{Kp}{Kp + (n-K)q}$ if $\RootLabel =1$ and is Bernoulli-$\frac{K}{n}$ if $\RootLabel =0$.
\end{itemize}
\end{proof}

The inference problem defined on the random tree is coupled to the recovering of a hidden community with side information through a coupling lemma~\cite{recovering_one}, which shows that under certain conditions, the neighborhood of a fixed node $i$ in the graph is locally a tree with probability converging to one. Thus, the belief propagation algorithm defined for random trees can be used on the graph as well. The proof of the coupling lemma depends only on the tree structure, implying that it also holds for our system model where the side information is independent of the tree structure given the labels. 

Define $\boldsymbol{G}_{u}^{\TreeIter}$ to be the subgraph containing all nodes that are at a distance at most $\TreeIter$ from node $u$ and define $\boldsymbol{x}_{u}^{\TreeIter}$ and $\boldsymbol{Y}_{u}^{\TreeIter}$ to be the set of labels and side information of all nodes in $\boldsymbol{G}_{u}^{\TreeIter}$, respectively. 

\begin{Lemma}[Coupling Lemma~\cite{recovering_one}]
\label{couple}
Suppose that $\TreeIter(n)$ are positive integers such that  $(2+np)^{\TreeIter(n)} = n^{o(1)}$. Then, for any node $u$ in the graph, there exists a coupling between $(\boldsymbol{G},\boldsymbol{x},\boldsymbol{Y})$ and $(T, \TreeLabels, \TreeSideInfo)$ such that:
\begin{equation}
\label{couple.eq.1}
\mathbb{P}( (\boldsymbol{G}_{u}^{\TreeIter}, \boldsymbol{x}_{u}^{\TreeIter},\boldsymbol{Y}_{u}^{\TreeIter} ) = (T^{\TreeIter}, \TreeLabels^{\TreeIter}, \TreeSideInfo^{\TreeIter}) ) \geq 1 - n^{-1 + o(1)}
\end{equation}
where for convenience of notation, the dependence of $\TreeIter$ on $n$ is made implicit.
\end{Lemma}

Now, we are ready to present the belief propagation algorithm for community recovery with bounded side information. Define the message transmitted from node $i$ to its neighboring node $j$ at iteration $t+1$ as:
\begin{align}
\label{BP.1}
R_{i \to j}^{t+1} = h_{i} - K(p-q) + \sum_{k \in {\mathcal N}_i\backslash j} M(R_{k \to i}^{t})
\end{align}
where $h_{i} = \log(\frac{\mathbb{P}(y_{i}|x_{i}=1)}{\mathbb{P}(y_{i}|x_{i}=0)})$, ${\mathcal N}_{i}$ is the set of neighbors of node $i$ and $M(x) = \log(\frac{\frac{p}{q}e^{x-\nu} + 1}{e^{x-\nu}+1})$. The messages are initialized to zero for all nodes $i$, i.e., $R_{i \to j}^{0} = 0$ for all $i \in \{ 1,\cdots,n\}$ and $j \in {\mathcal N}_i$. Define the belief of node $i$ at iteration $t+1$ as:
\begin{align}
\label{BP.2}
R_{i}^{t+1} = h_{i} - K(p-q) + \sum_{k \in {\mathcal N}_i} M(R_{k \to i}^{t})
\end{align}

Algorithm~\ref{BP.One.alg} presents the proposed belief propagation algorithm for community recovery with side information.
\begin{table*}
\caption{Belief propagation algorithm for community recovery with side information.}
\label{BP.One.alg}
\centering
\normalsize
\begin{tabular}{|p{4.5in}|}
\hline
Belief Propagation Algorithm \\
\hline
\begin{enumerate}
    \item Input: $n, K, \TreeDepth\in\mathbb{N}$, $\boldsymbol{G}$ and $\boldsymbol{Y}$.
    \item For all nodes $i$ and $j \in {\mathcal N}_i$, set $R^{0}_{i\to j} = 0$.
    \item For all nodes $i$ and $j \in {\mathcal N}_i$, run $\TreeDepth-1$ iterations of belief propagation as in~\eqref{BP.1}.
    \item For all nodes $i$, compute its belief $R_{i}^{\TreeDepth}$ based on~\eqref{BP.2}.
    \item Output $\tilde{C}=\{ \text{Nodes corresponding to $K$ largest }R_{i}^{\TreeDepth}\}$.
\end{enumerate}
\\
\hline
\end{tabular}
\end{table*}

If in Algorithm~\ref{BP.One.alg} we have $\TreeDepth=\TreeIter(n)$, according to Lemma~\ref{couple} with probability converging to one $R_{i}^{\TreeDepth} = \Gamma_{i}^{\TreeDepth}$, where $\Gamma_{i}^{\TreeDepth}$ was the log-likelihood defined for the random tree. Hence, the performance of Algorithm~\ref{BP.One.alg} is expected to be the same as the MAP estimator defined as $\hat{\tau}_{MAP} = 1_{\{\Gamma_{i}^{\TreeDepth} \geq \nu\}}$, where $\nu = \log(\frac{n-K}{K})$.


We now study the asypmtotic behaviour of $\Gamma_{i}^{\TreeDepth}$. Define for $ t\geq 1$ and any node $i$:
\begin{align}
\psi_{i}^{t} &\triangleq -K(p-q) +  \sum_{j \in {\mathcal N}_i} M(h_{j} + \psi_{j}^{t-1}) \label{psi.1}
\end{align} 
where 
\[
M(x) \triangleq \log\Big(\frac{\frac{p}{q}e^{x-\nu} + 1}{e^{x-\nu}+1}\Big) = \log\Big(1+ \frac{\frac{p}{q}-1}{1 + e^{-(x-\nu)}}\Big).
\]
Then, $\Gamma_{i}^{t+1} = h_{i} + \psi_{i}^{t+1}$ and $\psi_{i}^{0} = 0$ $\forall i \in \CroppedTree$. Let $Z_{0}^{t}$ and $Z_{1}^{t}$ denote random variables drawn according to the distribution of $\psi_{i}^{t}$ conditioned on $x_i= 0$ and $x_i = 1$, respectively. Similarly, let $U_{0}$ and $U_{1}$ denote  random variables drawn according to the distribution of $h_{i}$ conditioned on $\tau_i=0$ and $\tau_i = 1$, respectively.

\begin{Lemma} 
\label{mean_variance_BP}
(\cite[Lemma 11]{Saad:single}) Assume $\lambda$, $\frac{\alpha_{+,m}}{\alpha_{-,m}}$ and $\nu$ are constants independent of $n$ while $nq, Kq \overset{n\rightarrow\infty}{\xrightarrow{\hspace{0.2in}}} \infty$. Then, for all $t \geq 0$:
\begin{align}
\mathbb{E}[Z_{0}^{t+1}] &= \frac{-\lambda}{2}b_{t} + o(1) \label{mean_1} \\
\mathbb{E}[Z_{1}^{t+1}] &= \frac{\lambda}{2}b_{t} + o(1) \label{mean_2}\\
\text{var}(Z_{0}^{t+1}) &= \text{var}(Z_{1}^{t+1}) = \lambda b_{t} + o(1) \label{var_1}
\end{align}
\end{Lemma}


The following lemma shows that the distributions of $Z_{1}^{t}$ and $Z_{0}^{t}$ are asymptotically Gaussian. 

\begin{Lemma}
\label{Gaussian}
(\cite[Lemma 12]{Saad:single}) Assume $\lambda$, $\frac{\alpha_{+,m}}{\alpha_{-,m}}$ and $\nu$ are constants independent of $n$ while $nq, Kq \overset{n\rightarrow\infty}{\xrightarrow{\hspace{0.2in}}} \infty$. Let $\phi(x)$ be the cumulative distribution function (CDF) of a standard normal distribution. Define $v_{0} =0$ and $v_{t+1} =\lambda \mathbb{E}_{Z,U_{1}}[\frac{1}{e^{-\nu} + e^{-(\frac{v_{t}}{2} +\sqrt{v_{t}}Z)-U_{1}}}]$, where $Z \sim \mathcal{N}(0,1)$. Then, for all $t\geq 0$:
\begin{align}
&\sup_{x} \big| \mathbb{P}\big( \frac{Z_{0}^{t+1} + \frac{v_{t+1}}{2} }{\sqrt{v_{t+1}}} \leq x \big) - \phi(x) \big|  \to 0 \label{gauss.1} \\
&\sup_{x} \big| \mathbb{P}\big( \frac{Z_{1}^{t+1} - \frac{v_{t+1}}{2} }{\sqrt{v_{t+1}}} \leq x \big) - \phi(x) \big|  \to 0 \label{gauss.2}
\end{align}
\end{Lemma}


The following lemma characterizes the asymptotic residual error of belief propagation with side information for recovering a single community.

\begin{Lemma}
\label{MAP_unbounded}
Assume $\lambda$, $\frac{\alpha_{+,m}}{\alpha_{-,m}}$ and $\nu$ are constants independent of $n$ while $nq, Kq \overset{n\rightarrow\infty}{\xrightarrow{\hspace{0.2in}}} \infty$.  Let $\hat{C}$ define the community recovered by the MAP estimator, i.e. $\hat{C} = \{i: \Gamma_{i}^{t} \geq \nu \}$. Then,
\begin{align}
\lim_{nq,Kq \to \infty} \lim_{n\to\infty} \frac{\mathbb{E}[\hat{C} \triangle C^{*}]}{K} &= \frac{n-K}{K} \mathbb{E}_{U_{0}}[Q(\frac{\nu+\frac{v_{t}}{2} - U_{0}}{\sqrt{v_{t}}})] \twocolbreak + \mathbb{E}_{U_{1}}[Q(\frac{-\nu+\frac{v_{t}}{2} + U_{1}}{\sqrt{v_{t}}})] \label{main.one.SBM}
\end{align}
where $v_{0} =0$ and $v_{t+1} =\lambda \mathbb{E}_{Z,U_{1}}[\frac{1}{e^{-\nu} + e^{-(\frac{v_{t}}{2} +\sqrt{v_{t}}Z)-U_{1}}}]$, and $Z \sim \mathcal{N}(0,1)$.
\end{Lemma}

\begin{proof}
Let $p_{e,0}, p_{e,1}$ denote Type I and Type II errors for recovering $\tau_0$. Then, the proof follows from Lemmas~\ref{mean_variance_BP} and~\ref{Gaussian}, using
\[
\frac{\mathbb{E}[\hat{C} \triangle C^{*}]}{K} = \frac{n}{K} p_{e}^{t}  = \frac{n-K}{K}  p_{e,0} + p_{e,1}.
\]
\end{proof}


\subsection{Exit Analysis}\label{EX.1}

An interesting and natural question is: does belief propagation with side information have a phase transition? If yes, what is the threshold? Equation~\eqref{main.one.SBM} shows the residual asymptotic error of belief propagation for detecting one community with side information. However, it does not provide a direct answer regarding phase transition. This section demonstrates the utility of EXIT charts in the understanding of phase transition.

We begin by calculating the mutual information between the label of node $i$, $x_{i}$, and its belief at time $t$, $R_{i}^{t}$ as follows:
\begin{align}
I&(x_{i},R_{i}^{t}) \nonumber \\
= & -\frac{K}{n}\log(\frac{K}{n}) - (1-\frac{K}{n})\log(1-\frac{K}{n}) - H(x_{i}|R_{i}^{t}) \nonumber \\
= & -\frac{K}{n}\log(\frac{K}{n}) - (1-\frac{K}{n})\log(1-\frac{K}{n})    \nonumber \\
&-\frac{K}{n} \int_{-\infty}^{\infty} \Bigg(\sum_{m=1}^{M} \alpha_{+,m} \frac{e^{\frac{-(y-(v_{t}+ h_{m}))^{2}}{2v_{t}}}}{\sqrt{2\pi v_{t}}}\Bigg) \twocolbreak{} \log_{2}\bigg(1 + \frac{(n-K)\sum_{m=1}^{M} \alpha_{-,m} e^{\frac{-(y-(-v_{t}+ h_{m}))^{2}}{2v_{t}}}}{K \sum_{m=1}^{M} \alpha_{+,m} e^{\frac{-(y-(v_{t}+ h_{m}))^{2}}{2v_{t}}}} \bigg) dy   \nonumber \\
& -\frac{n-K}{n} \int_{-\infty}^{\infty} \Bigg(\sum_{m=1}^{M} \alpha_{-,m} \frac{e^{\frac{-(y-(-v_{t}+ h_{m}))^{2}}{2v_{t}}}}{\sqrt{2\pi v_{t}}}\Bigg) \twocolbreak{} \log_{2}\Bigg(1 + \frac{K \sum_{m=1}^{M} \alpha_{+,m} e^{\frac{-(y-(v_{t}+ h_{m}))^{2}}{2v_{t}}}}{(n-K) \sum_{m=1}^{M} \alpha_{-,m} e^{\frac{-(y-(-v_{t}+ h_{m}))^{2}}{2v_{t}}}} \Bigg) dy  \label{main.Mutual.One.Comm.}
\end{align}
where $h_m = \log(\frac{u_{+,m}}{u_{-,m}})$. 

For a concrete demonstration of the capabilities of EXIT analysis, we use the following model for side information. Let $M=2$, where for each node $i$, $y_{i} = x_{i}$ with probability $1-\alpha$, and $y_i = 1-x_i$ with probability $\alpha$, where $\alpha \in [0,0.5]$.  
Note that for a fixed $\alpha$, $I(x_{i},R_{i}^{t})$ is function of $v_{t}$ only. Hence, we will denote it by $J(v_{t})$.

\begin{figure*}
\centering
\begin{minipage}{3.209in}
\centering
\includegraphics[width=3.209in]{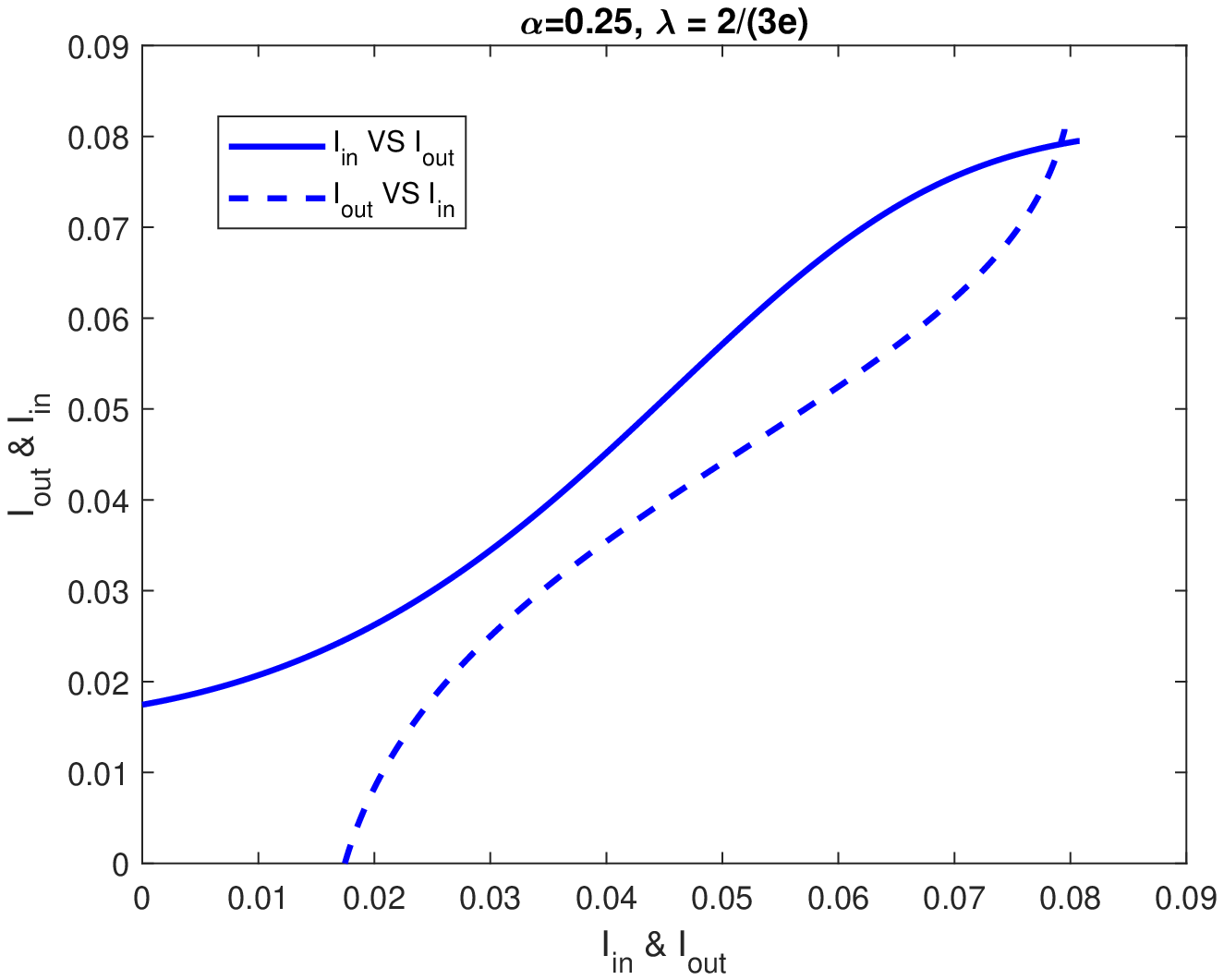}\\
\end{minipage}
\begin{minipage}{3.209in}
\centering 
\includegraphics[width=3.209in]{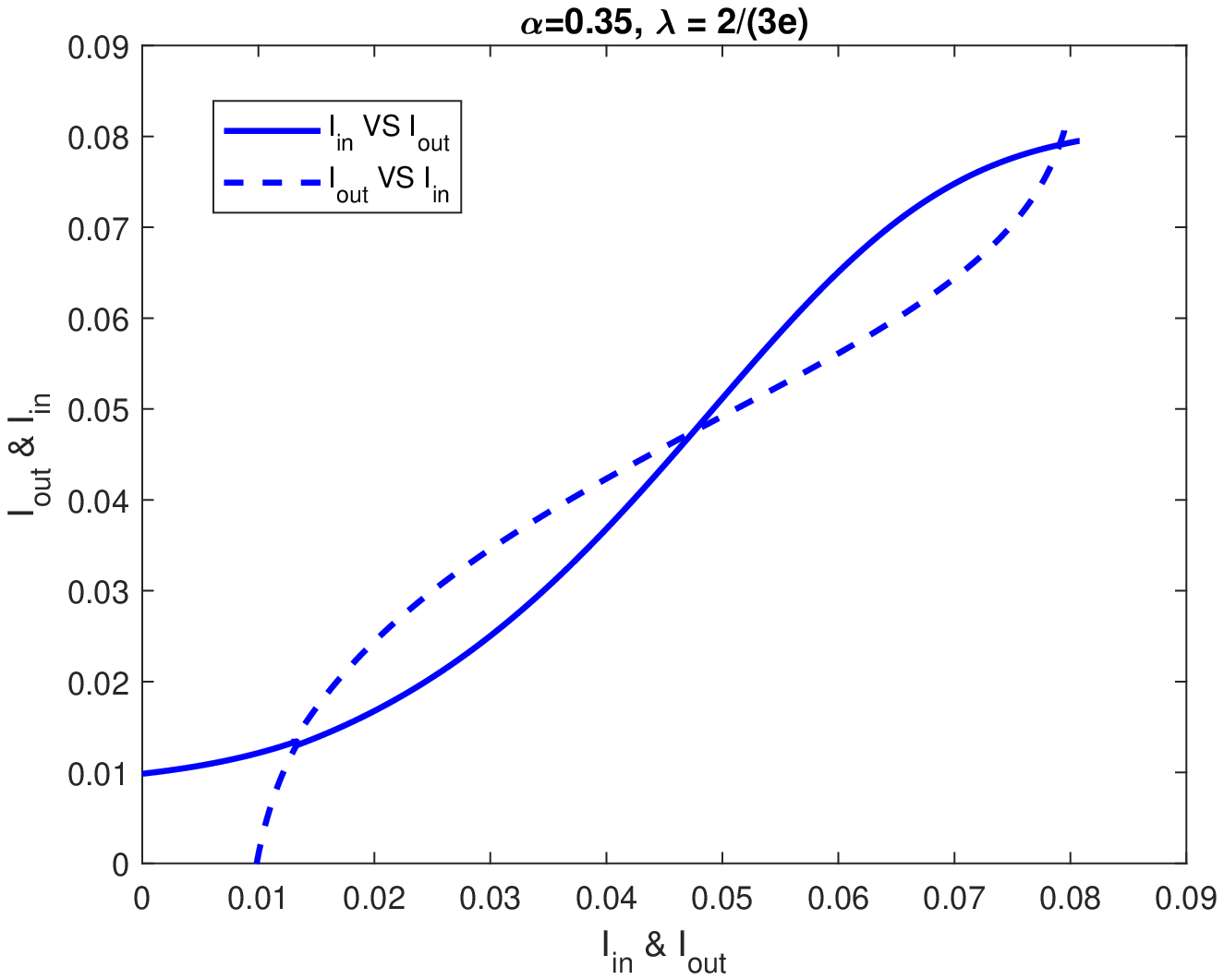}\\
\end{minipage}
\\[0.2in]
\centering
\begin{minipage}{3.209in}
\centering
\includegraphics[width=3.209in]{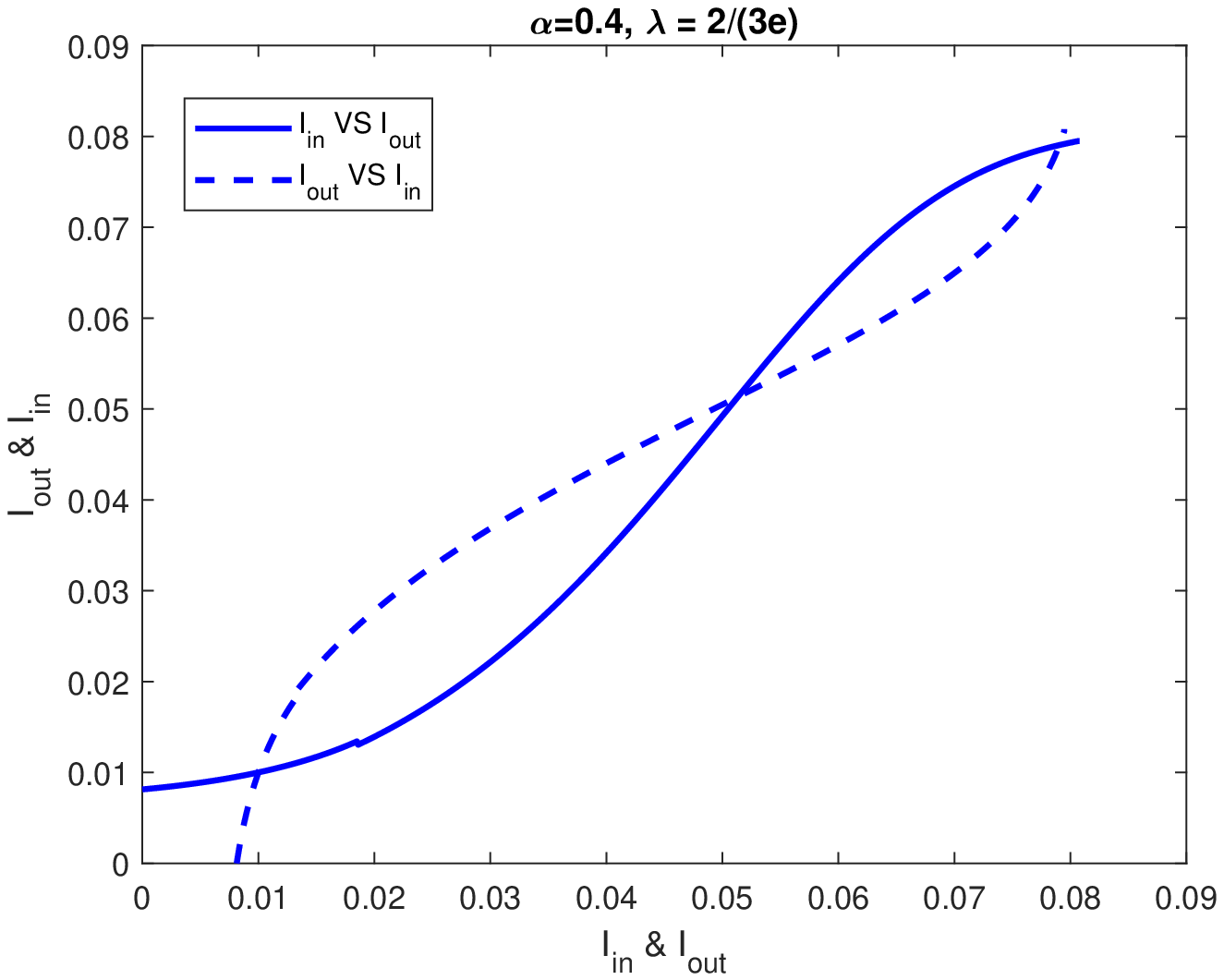}\\
\end{minipage}
\begin{minipage}{3.209in}  
\centering 
\includegraphics[width=3.209in]{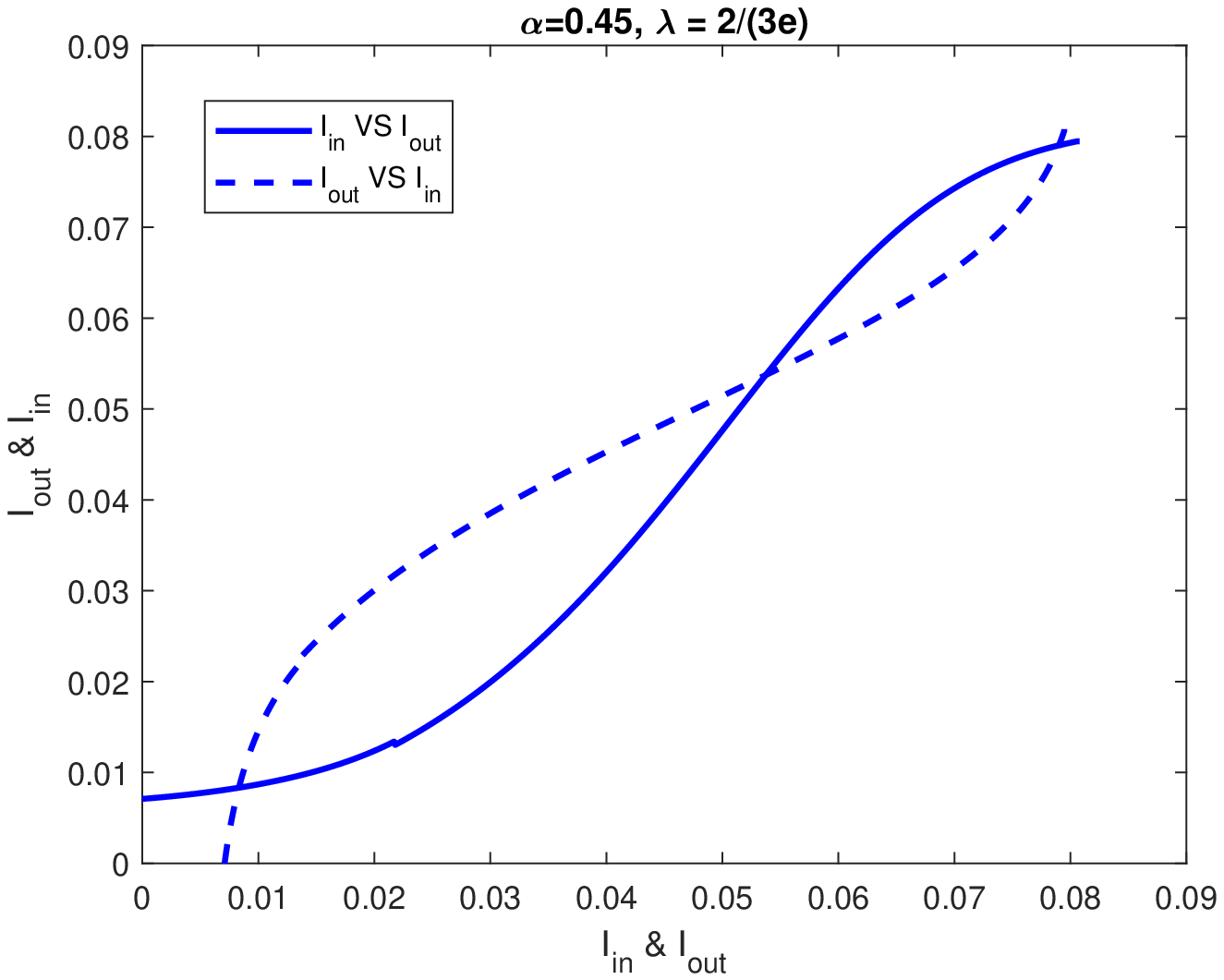}\\
\end{minipage}
\caption{EXIT charts for one community detection with $\lambda = \frac{2}{3e}$ for different values of $\alpha$.} 
\label{fig1.0}
\end{figure*}

Based on the belief propagation algorithm described in Table~\ref{BP.One.alg}, at iteration $t$, node $i$ receives the beliefs of all nodes $j \in \mathcal{N}_i$ calculated at iteration $(t-1)$. We denote the input information to node $i$ from node $j$ as $I_{in}$. Then, node $i$ computes the new information it has at iteration $t$, which we call $I_{out}$. Note that $I_{in}$ and $I_{out}$ can be calculated using~\eqref{main.Mutual.One.Comm.} as $J(v_{t-1})$ and $J(v_{t})$, respectively. Since $J(v_{t})$ is monotonically increasing in $v_{t}$~\cite{exit}, $J(v_{t})$ is reversible. Thus, $v_{t} = J^{-1}(I(x_{i},R_{i}^{t}))$.  Moreover, since $v_{t-1}$ and $v_{t}$ are related by: 
\[v_{t} =\lambda \mathbb{E}_{Z,U_{1}}\bigg[\frac{1}{e^{-\nu} + e^{-(\frac{v_{t-1}}{2} +\sqrt{v_{t-1}}Z)-U_{1}}}\bigg],
\]
we can define the relation between $I_{in}$ and $I_{out}$ for node i as follows:
\begin{align}
\twocolAlignMarker I_{out} =  \twocolnewline &  J\Bigg(\lambda \bigg[ \alpha \mathbb{E}_{Z}\Big[\Big(e^{-\nu} + e^{-(\frac{J^{-1}(I_{in})}{2} +\sqrt{J^{-1}(I_{in})}Z)-\log(\frac{\alpha}{1-\alpha}) }\Big)^{-1}\Big] + \nonumber \\
	& (1-\alpha)\mathbb{E}_{Z}\Big[\Big(e^{-\nu} + e^{-(\frac{J^{-1}(I_{in})}{2} +\sqrt{J^{-1}(I_{in})}Z)-\log(\frac{1-\alpha}{\alpha}) }\Big)^{-1}\Big] \bigg] \Bigg) 
\end{align}
To compute $J$ and $J^{-1}$, we apply curve fitting using the Levenberg Marquardt algorithm~\cite{exit}.

Figures~\ref{fig1.0},~\ref{fig1.1} and~\ref{fig2.1} show the EXIT curves for different values of $\lambda$, and $\alpha$. 
\begin{itemize}
\item Figure~\ref{fig1.0} shows a threshold for $\lambda$, such that the EXIT curves do not intersect above this threshold and they do below the threshold. Hence, belief propagation with side information experiences a phase transition.  Moreover, above the threshold, the maximum mutual information is attained, and hence, a vanishing residual error is possible (weak recovery). This particular example is constructed for a graph whose probability distribution does not provide sufficient information {\em alone} for weak recovery. This example demonstrates clearly the role of side information in weak recovery especially in conditions where, without it, weak recovery is not attainable. EXIT analysis thus confirms the threshold effect that was first reported in~\cite{Saad:single}, but more importantly, EXIT demonstrates the phase transition behavior in a visually compelling manner that is easy to grasp, with relatively straight forward calculations. 

\item
To elaborate, EXIT charts bring further clarity to the nature of the belief propagation threshold, by showing how the iterations of the belief propagation, at threshold, just barely manage to escape through a bottleneck and approach the maximum likelihood solution. EXIT also clearly demonstrates the residual error of belief propagation on the two sides of the phase transition (the jump in error probability at phase transition) which is not as easy to see via other analytical methods. 

\item
Thus, the EXIT method demonstrates that while the thresholding phenomenon for belief propagation is indeed sharp in terms of transition across parameters of the model for the graph and side information, however, close to the threshold the belief propagation might pay a heavy price in terms of the number of iterations needed to converge. Thus, in the sense of the cost of the algorithm, the behavior of belief propagation near the threshold is something that is especially well understood via the EXIT analysis. The curvature (second derivative) of the EXIT curves at the point of bottleneck is an indication of the iterations needed close to the threshold. This effect is not visible to the other analytical methods that, typically, first let the number of iterations go to infinity, and then observe the (asymptotic, in iterations) performance of the belief propagation algorithm across the landscape of the parameters of the system model.

\item
As mentioned earlier, Fig.~\ref{fig1.0} shows the thresholding effect for the side information where the graphical information is fixed. In order to complete the picture, we also performed experiments where we hold the quality of the side information to be fixed (via a fixed $\alpha$), while we allow the graph to become progressively more informative (characterized by improving $\lambda$). This result is shown in Figures~\ref{fig1.1} and~\ref{fig2.1}. In these figures, the thresholding effect for graphical information is shown in the presence of side information.




\end{itemize}

\begin{figure}
\centering
\includegraphics[width=\Figwidth]{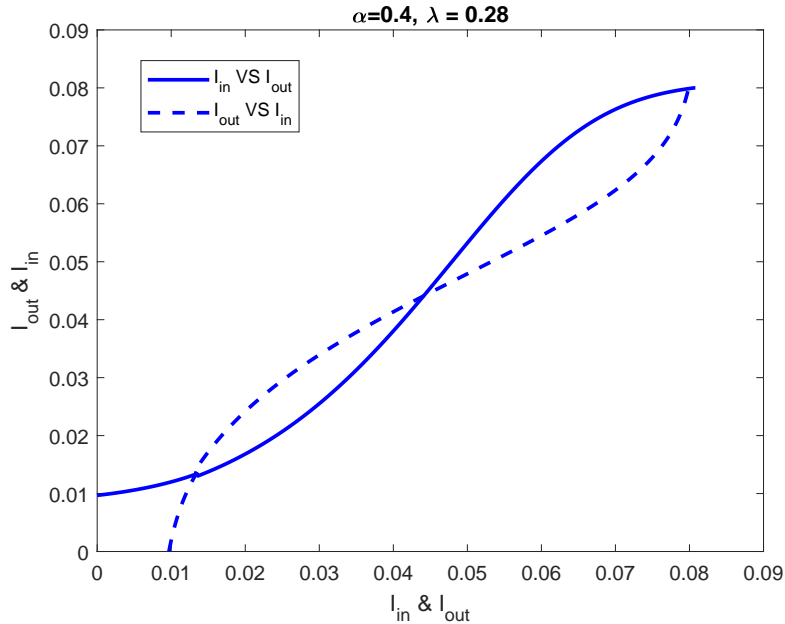}
\caption{EXIT Chart for one community detection with $\alpha = 0.4$.}
\label{fig1.1}
\end{figure}

\begin{figure}
\centering
\includegraphics[width=\Figwidth]{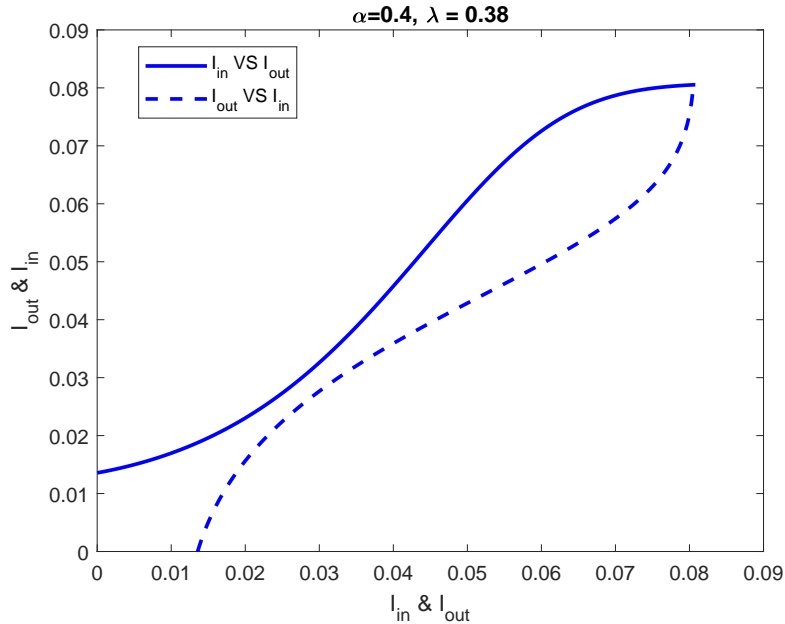}
\caption{EXIT Chart for one community detection with $\alpha = 0.4$.}
\label{fig2.1}
\end{figure}

%
%
%
%
%


\section{Conclusion}
\label{con.}

This paper proposes the extrinsic information transfer (EXIT) method to study the performance of belief propagation in community detection under side information. The EXIT technique was introduced originally for the analysis of turbo codes and LDPC codes. We consider the stochastic block model for one community and two symmetric communities. For single community detection, we demonstrate the suitability of EXIT analysis to predict whether belief propagation experiences a phase transition, and where is it. Furthermore, we show the power of EXIT analysis to produce insights that are not easily available otherwise, such as the performance and complexity of belief propagation near the threshold. For the two symmetric communities, we calculate the asymptotic residual error for belief propagation with finite-alphabet side information, generalizing a result in the literature. This work shows that EXIT analysis can illuminate  the effect of quality and quantity of side information on the performance of belief propagation in terms of residual error, performance through the first few iterations, and achieving weak recovery. 


\bibliographystyle{IEEEtran}
\bibliography{ISIT2016}

\end{document}